\documentclass[12pt,english,draftcls, one column]{IEEEtran}

\usepackage[T1]{fontenc}
\usepackage[latin9]{inputenc}
\usepackage{float}
\usepackage{amsmath}
\usepackage{amsthm}
\usepackage{amssymb}
\usepackage{graphicx}
\usepackage{stfloats}
\usepackage{cite} 
\usepackage{xcolor}
\usepackage{bbm}
\usepackage{multirow}
\usepackage{multicol}
\usepackage{subfig}

\usepackage{algorithm}
\usepackage{algpseudocode}

\makeatletter

\floatstyle{ruled}
\newfloat{algorithm}{tbp}{loa}
\providecommand{\algorithmname}{Algorithm}
\floatname{algorithm}{\protect\algorithmname}

\theoremstyle{plain}
\newtheorem{thm}{\protect\theoremname}
\theoremstyle{plain}
\newtheorem{prop}[thm]{\protect\propositionname}

\IEEEoverridecommandlockouts
\usepackage{cite}
\usepackage{amsfonts} 
\usepackage{textcomp}
\usepackage{etoolbox}

\def\BibTeX{{\rm B\kern-.05em{\sc i\kern-.025em b}\kern-.08em
    T\kern-.1667em\lower.7ex\hbox{E}\kern-.125emX}}

\patchcmd{\@maketitle}
  {\addvspace{0.5\baselineskip}\egroup}
  {\addvspace{-1\baselineskip}\egroup}
  {}
  {}

\makeatother

\providecommand{\propositionname}{Proposition}
\providecommand{\theoremname}{Theorem}

\begin{document}

\title{\huge Multi-Satellite Cooperative MIMO Transmission: \\Statistical CSI-Aware RSMA Precoding Design}

\author{Sangwon Jo, \IEEEmembership{Graduate Student Member, IEEE}, and \\ Seok-Hwan Park, \IEEEmembership{Senior Member, IEEE} \thanks{This work was supported in part by the Regional Innovation System $\&$ Education (RISE) initiative, funded by the Ministry of Education (MOE) and administered by the National Research Foundation (NRF) of Korea; in part by the NRF funded by the MOE under Grant RS-2019-NR040079; and in part by NRF funded by the MSIT under Grant RS-2023-00238977.}
\thanks{The authors are with the Division of Electronic Engineering, Jeonbuk National University, Jeonju 54896, Korea (email: tkddnjs9803@jbnu.ac.kr and seokhwan@jbnu.ac.kr).
}
}

\maketitle
\begin{abstract}
We investigate inter-satellite cooperative transmission in a multiple low-Earth orbit (LEO) satellite communication system to enhance spectral efficiency. Specifically, we design multiple-input multiple-output (MIMO) precoding at LEO satellites for cooperative rate-splitting multiple access (RSMA). Given the difficulty of acquiring instantaneous channel state information (iCSI) due to long delays and Doppler effects, we formulate an ergodic max-min fairness rate (MMFR) maximization problem based on statistical CSI (sCSI).
To address the challenge of ergodic rate evaluation, we approximate the problem using closed-form upper bounds and develop a weighted minimum mean squared error-based algorithm to obtain a stationary point.
Simulation results demonstrate that the proposed sCSI-based RSMA scheme approaches iCSI-based performance and significantly outperforms conventional space-division multiple access.
\end{abstract}


\theoremstyle{theorem}
\newtheorem{theorem}{Theorem} 
\theoremstyle{proposition}
\newtheorem{proposition}{Proposition} 
\theoremstyle{lemma}
\newtheorem{lemma}{Lemma} 
\theoremstyle{corollary}
\newtheorem{corollary}{Corollary} 
\theoremstyle{definition}
\newtheorem{definition}{Definition}
\theoremstyle{remark}
\newtheorem{remark}{Remark}

\section{Introduction} \label{sec:intro}

Massive multiple-input multiple-output (mMIMO) utilizes a large number of antennas to deliver reliable and high data-rate communication links \cite{Bjornson-et-al:TWC20}.
Several recent works have investigated mMIMO transmission for satellite communication (SATCOM) using low-Earth-orbit (LEO) satellites \cite{You:JSAC20, Li:TC22, You:JSAC22, Xiang:TC24, Xiang:TAES24, Zhang:TWC24}.
In particular, reference \cite{You:JSAC20} studied the channel modeling and designed downlink precoders and uplink receivers based on statistical channel state information (sCSI), recognizing that acquiring instantaneous CSI (iCSI) is highly challenging in LEO SATCOM due to impairments such as propagation delay and Doppler shifts, whereas sCSI remains stable over longer time scales.
In \cite{Li:TC22}, it was shown that single-stream precoding maximizes the ergodic sum rate even when each user terminal (UT) is equipped with multiple antennas, owing to the rank-one nature of the channel matrices between each LEO satellite and UT.
In line with the growing density of LEO satellite constellations \cite{Abdelsadek:TC23}, mMIMO transmission for multi-LEO SATCOM systems has been addressed in \cite{Xiang:TC24, Xiang:TAES24, Zhang:TWC24}.
While references \cite{Xiang:TC24} and \cite{Xiang:TAES24} considered coordinated transmission and reception across satellites, where only CSI-related parameters are exchanged among satellites, reference \cite{Zhang:TWC24} investigated a fully-cooperative transmission scheme in which multiple LEO satellites exchange both CSI and UTs' data via inter-satellite links (ISLs).

While the aforementioned works focused on conventional space-division multiple access (SDMA) techniques, the integration of advanced rate-splitting multiple access (RSMA) technique \cite{Park:Net23} into SATCOM systems was explored in \cite{Yin:TC21, Xu:TC24, Liu:TWC24}.
Specifically, \cite{Yin:TC21} addressed RSMA-based multigroup multicast beamforming in multi-beam SATCOM systems under imperfect CSI, \cite{Xu:TC24} investigated joint geostationary orbit (GEO) and LEO precoding and message splitting in a hybrid GEO-multiple LEO network, and \cite{Liu:TWC24} focused on energy-efficient RSMA design in ISAC-enabled LEO satellite systems. However, these works did not consider cooperative transmission among multiple LEO satellites \cite{Yin:TC21, Xu:TC24, Liu:TWC24}, nor did they address practical optimization with sCSI in light of challenges in acquiring iCSI \cite{Yin:TC21, Liu:TWC24}.

Motivated by these insights, this work investigates the use of RSMA for cooperative downlink mMIMO transmission in a multiple-satellite SATCOM system, leveraging ISLs. Our objective is to optimize the downlink precoder to maximize the ergodic max-min fairness rate (MMFR) among UTs, relying solely on sCSI.
To circumvent the computational burden of estimating the ergodic rate via sample averaging, we reformulate the problem using a closed-form upper bound on the ergodic rate. Given the non-convexity of the resulting optimization problem, we develop a  weighted minimum mean square error (WMMSE)-based algorithm \cite{Christensen:TWC08}. Numerical results demonstrate that the proposed RSMA-based cooperative scheme significantly outperforms conventional SDMA and non-cooperative baselines, achieving performance close to that with perfect iCSI.

Main contributions of this work are summarized as follows.
\begin{itemize}
    \item We consider a multi-LEO SATCOM system with ISLs, enabling cooperative mMIMO RSMA transmission for effective interference management.
    \item Recognizing the challenges in acquiring iCSI, we leverage sCSI, which remains stable over longer periods. We formulate the problem of optimizing the downlink precoders to maximize the ergodic MMFR among UTs.
    \item To overcome the challenge of computing ergodic rates, we reformulate the problem using closed-form upper bounds. We then develop a WMMSE-based algorithm, which guarantees a stationary point.
    \item Through numerical evaluations, we demonstrate that the proposed sCSI-based RSMA scheme achieves performance close to that of its iCSI-based counterpart, outperforming conventional SDMA and non-cooperative baseline schemes.
\end{itemize}

\section{System Model\label{sec:System-Model}}

\begin{figure}
\centering\includegraphics[width=1.0\linewidth]{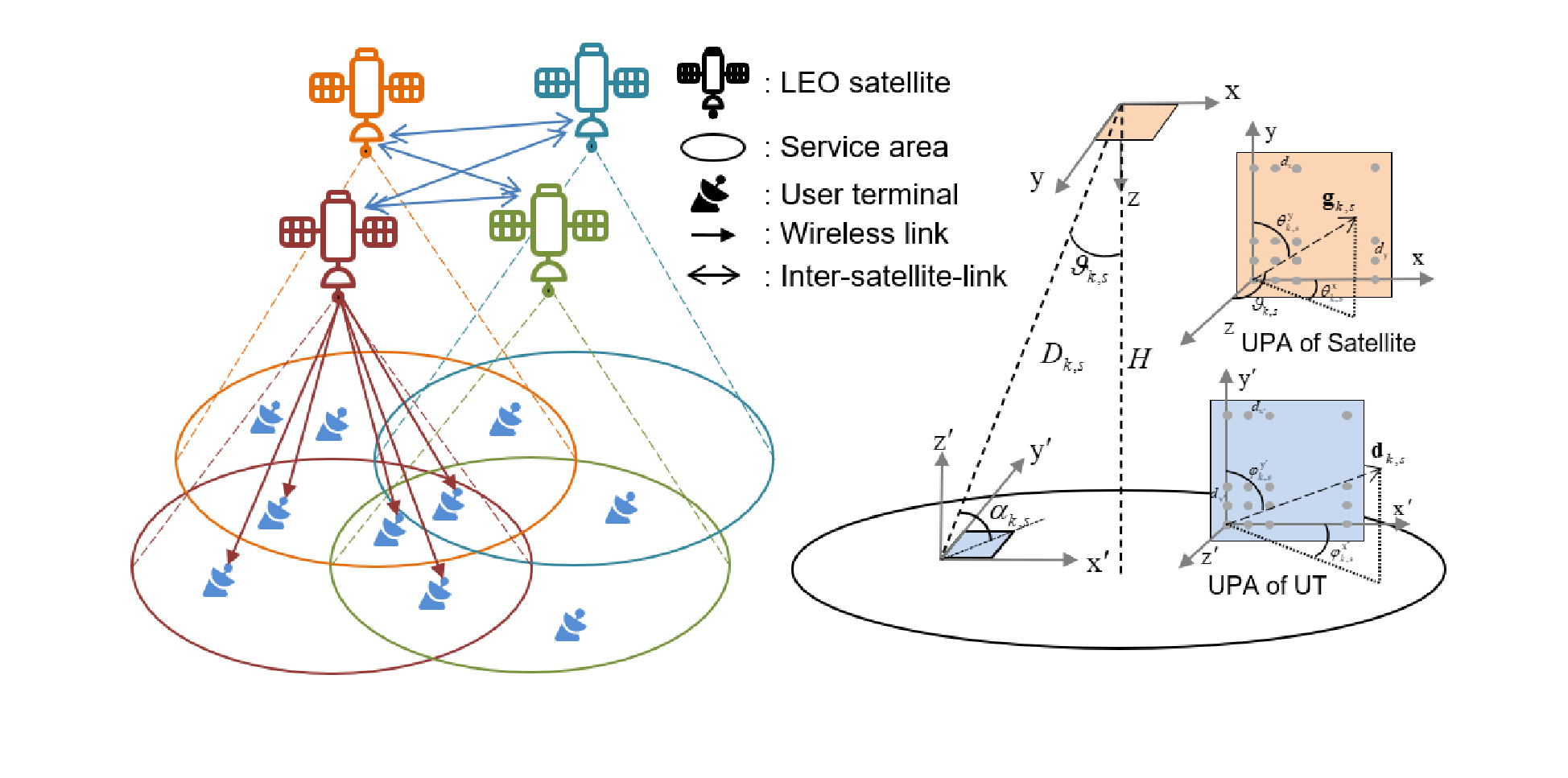}
\caption{\label{fig:System-model} Illustration of the downlink of multiple-LEO satellite communication system with $S=4$ and $K=10$.}
\end{figure}

We consider a downlink transmission system in an NTN comprising multiple LEO satellites, cooperating with each other via ISLs, and ground UTs as shown in Fig. \ref{fig:System-model}.
Specifically, we consider $S$ LEO satellites at altitude $H$ and $K$ UTs on the ground.
Each satellite is equipped with a uniform planar array (UPA) of $M = M_{\text{x}}M_{\text{y}}$ antenna elements, where $M_{\text{x}}$ and $M_{\text{y}}$ represent the number of antennas along the $\text{x}$-axis and $\text{y}$-axis, respectively.
Similarly, each UT is equipped with a UPA consisting of $N=N_{\text{x}{'}}N_{\text{y}{'}}$ elements, where $N_{\text{x}{'}}$ and $N_{\text{y}{'}}$ denote the number of antennas along the $\text{x}{'}$-axis and the $\text{y}{'}$-axis.
We define the index sets of satellites and UTs as $\mathcal{S} = \{1,2,\ldots,S\}$ and $\mathcal{K} = \{1,2,\ldots,K\}$, respectively.
The subset of satellites serving UT $k \in \mathcal{K}$ is denoted by $\mathcal{S}_k$, and the subset of UTs served by satellites $s \in \mathcal{S}$ is denoted by $\mathcal{K}_s$.
Due to potential overlaps in satellite coverage areas, as depicted in Fig. \ref{fig:System-model}, we assume cooperative transmission among satellites to jointly serve UTs within these overlapping regions.

\subsection{Satellite-to-UT Channel Model} \label{sub:Tx-signal-channel}

The received signal $\mathbf{y}_k \in \mathbb{C}^{N \times 1}$ at each UT $k$ is given by
\begin{align}
    \mathbf{y}_k = \sum\nolimits_{s\in\mathcal{S}} \mathbf{H}_{k,s} \mathbf{x}_s + \mathbf{z}_k = \mathbf{H}_k \mathbf{x} + \mathbf{z}_k, \label{eq:received-signal-UT-k}
\end{align}
where $\mathbf{H}_{k,s} \in \mathbb{C}^{N \times M}$ denotes the channel matrix from satellite $s$ to UT $k$, $\mathbf{x}_s \in \mathbb{C}^{M \times 1}$ is the transmit signal vector from satellite $s$, and $\mathbf{z}_k\sim \mathcal{CN}(\mathbf{0}, \sigma_z^2\mathbf{I})$ is the additive noise vector. 
We define the aggregated channel matrix $\mathbf{H}_k\in\mathbb{C}^{N\times MS}$ and the stacked transmitted signal vector $\mathbf{x}\in\mathbb{C}^{MS\times 1}$ across all the satellites as $\mathbf{H}_k=[\mathbf{H}_{k,1} \dotsi \mathbf{H}_{k,S}]$ and $\mathbf{x}=[\mathbf{x}_1^H \dotsi \mathbf{x}_S^H]^H$, respectively.

Following \cite{Xiang:TC24}, the channel matrix $\mathbf{H}_{k,s}$ is expressed as
\begin{align}
    \mathbf{H}_{k,s} = \mathbf{d}_{k,s}\mathbf{g}_{k,s}^H, \label{eq:channel-matrix}
\end{align}
where $\mathbf{d}_{k,s}\in\mathbb{C}^{N\times 1}$ and $\mathbf{g}_{k,s}\in\mathbb{C}^{M\times 1}$ denote the receive and transmit array response vectors at UT $k$ and satellite $s$, respectively. 
The array response vectors for the $l$-th propagation path, denoted by $\mathbf{d}_{k,s,l}$ and $\mathbf{g}_{k,s,l}$, are defined as
$\mathbf{g}_{k,s,l} = \mathbf{e}_{\text{x}}(q_{{k,s,l}}^{\text{x}}) \otimes \mathbf{e}_{\text{y}}(q_{k,s,l}^{\text{y}})$ and $\mathbf{d}_{k,s,l} = \mathbf{e}_{\text{x}{'}}(q_{k,s,l}^{\text{x}{'}}) \otimes \mathbf{e}_{\text{y}{'}}(q_{k,s,l}^{\text{y}{'}})$, where the steering vector along axis $\text{v}\in\{\text{x},\text{y},\text{x}{'},\text{y}{'}\}$ is defined as  $\mathbf{e}_{\text{v}}(x)=\frac{1}{n_{\text{v}}}(1, e^{-j2\pi\frac{d_{\text{v}}}{\lambda}x},\ldots,e^{-j2\pi(n_{\text{v}}-1) \frac{d_{\text{v}}}{\lambda}x})^T \in \mathbb{C}^{n_{\text{v}} \times 1}$
with $q_{k,s,l}^{\text{v}}$ being the directional cosine of the $l$-th propagation path along axis $\text{v}$, $n_{\text{v}}$ the number of antennas, $d_{\text{v}}$ the antenna spacing, and $\lambda$ the carrier wavelength.
The directional cosines for the $l$-th path are given by $q_{k,s,l}^{\text{x}}=\sin{\theta_{k,s,l}^{\text{y}}} \cos{\theta_{k,s,l}^{\text{x}}},\, q_{k,s,l}^{\text{y}} = \cos{\theta_{k,s,l}^{\text{y}}}, \,q_{k,s,l}^{\text{x}{'}} = \sin{\varphi_{k,s,l}^{\text{y}{'}}} \cos{\varphi_{k,s,l}^{\text{x}{'}}},$ and $q_{k,s,l}^{\text{y}{'}} = \cos{\varphi_{k,s,l}^{\text{y}{'}}}$, where $\theta_{k,s,l}^{\text{x}},$ $\theta_{k,s,l}^{\text{y}},\,\varphi_{k,s,l}^{\text{x}{'}}$ and $\varphi_{k,s,l}^{\text{y}{'}}$ are the angles of the $l$-th path, as illustrated in Fig. \ref{fig:System-model}.
Given the considerable distance between LEO satellites and ground UTs, the multipath effects are predominantly shaped by the scattering environment near the UT.
Therefore, the directional cosines for different paths associated with the same UT can be approximated as constant, i.e., $q_{k,s,l}^{\text{x}} = q_{k,s}^{\text{x}}$ and $q_{k,s,l}^{\text{y}} = q_{k,s}^{\text{y}}$ \cite{You:JSAC20}.
As a result, the transmit array response simplifies to $\mathbf{g}_{k,s,l} = \mathbf{g}_{k,s}$, which varies slowly over time and is assumed to be known at satellite $s$ \cite{Xiang:TC24}.

On the other hand, the receive array response $\mathbf{d}_{k,s,l}$ may vary with the path index $l$, and the overall receive response vector $\mathbf{d}_{k,s}$ is modeled as 
\begin{align}
    \mathbf{d}_{k,s} = \sqrt{\frac{\beta_{k,s}\kappa_{k,s}}{\kappa_{k,s}+1}} \mathbf{d}_{k,s,0} + \sqrt{\frac{\beta_{k,s}}{\kappa_{k,s}+1}} \hat{\mathbf{d}}_{k,s},
\end{align}
where $\mathbf{d}_{k,s,0}$ is the line-of-sight (LoS) component derived from $\{\varphi_{k,s,0}^{\text{x}'},\varphi_{k,s,0}^{\text{y}'}\}$ satisfying \begin{align}
    \sin\left(\varphi_{k,s,0}^{\text{x}'}\right)\sin\left(\varphi_{k,s,0}^{\text{y}'}\right) = \sin \alpha_{k,s}.
\end{align}
The parameter $\beta_{k,s} =G_{\text{sat}}G_{\text{ut}} / (4\pi D_{k,s}/\lambda)^2$ denotes the average channel power 
with the antenna gains $G_{\text{sat}}$ and $G_{\text{ut}}$, the propagation distance $D_{k,s}$, and the carrier wavelength $\lambda$.
The Rician factor is denoted by $\kappa_{k,s} > 0$, and 
$\hat{\mathbf{d}}_{k,s} \sim \mathcal{CN}(\mathbf{0}, \mathbf{\Sigma}_{k,s})$ represents the non LoS (NLoS) component with covariance matrix $\mathbf{\Sigma}_{k,s} \succeq \mathbf{0}$ normalized such that $\text{tr}(\mathbf{\Sigma}_{k,s})=1$.

Based on the discussion above, the aggregated channel matrix $\mathbf{H}_{k}$ from all satellites to UT $k$ is given by
\begin{align}
    \mathbf{H}_k = \mathbf{D}_k \mathbf{G}_k^H,
\end{align}
where $\mathbf{D}_k=[\mathbf{d}_{k,1} \dotsi \mathbf{d}_{k,S}] \in \mathbb{C}^{N \times S}$, and $\mathbf{G}_k = \mathrm{blkdiag}(\{ \mathbf{g}_{k,s} \}_{s \in \mathcal{S}}) \in \mathbb{C}^{MS \times S}$.

We assume that each satellite $s$ independently estimates the local sCSI to its associated UTs, including the set $\{\beta_{k,s},\, \kappa_{k,s}, \, \mathbf{g}_{k,s}, \, \mathbf{d}_{k,s,0}, \, \mathbf{\Sigma}_{k,s}\}_{k \in \mathcal{K}_s}$.
Since this sCSI is determined by relatively stable parameters such as the satellite-UT geometry, we assume that it remains constant during the precoding design phase \cite{You:JSAC20}.
After estimating sCSI, each satellite subsequently shares its local sCSI with other cooperating satellites via ISLs.
Through this exchange, all participating $S$ satellites acquire access to the global sCSI required for cooperative transmission.
The delay due to sCSI exchange may make the shared sCSI outdated, leading to imperfect sCSI. A robust design with imperfect sCSI is left for future work.

\subsection{Cooperative RSMA Transmission} \label{sub:cooperative-transmission}

To effectively manage inter-UT interference, we adopt the RSMA framework \cite{Park:Net23}. 
Specifically, we define the data stream vector as $\mathbf{m} = [m_c, m_{p,1}, m_{p,2}, \dots , m_{p,K}]^T \in \mathbb{C}^{(K+1) \times 1},$ where $m_c \sim \mathcal{CN}(0,1)$ represents the common data stream intended for all UTs, and $m_{p,k} \sim \mathcal{CN}(0,1)$ denotes the private data stream intended exclusively for UT $k$.
The common stream $m_c$ decoded by all UTs enables RSMA to manage interference more flexibly than SDMA \cite{Park:Net23}, as both the beamforming vector for $m_c$ and the allocation of the common rate among UTs can be jointly optimized.
We assume the presence of reliable ISLs, which enables the LEO satellites to cooperatively transmit the data streams $\mathbf{m}$.

At each satellite $s$, the data stream vector $\mathbf{m}$ is linearly precoded into the transmit signal vector $\mathbf{x}_s$ using a precoding matrix $\mathbf{Q}_s=[\mathbf{q}_{c,s}, \mathbf{q}_{p,1,s}, \dots , \mathbf{q}_{p,K,s}] \in \mathbb{C}^{M \times (K+1)}$, where $\mathbf{q}_{c,s}$ and $\{ \mathbf{q}_{p,k,s}\}_{k\in\mathcal{K}}$ denote the common and private precoding vectors at satellite $s$, respectively.
Defining the precoding matrix $\mathbf{Q} \in \mathbb{C}^{MS \times (K+1)}$ stacked across all the satellites as $\mathbf{Q} = [\mathbf{Q}_{1}^H \dotsi \mathbf{Q}_{S}^H]^H$, the transmitted signal vector $\mathbf{x}$ across all satellites is given by
\begin{align}
    \textbf{x} = \mathbf{Q}\mathbf{m} = \mathbf{q}_{c}m_c + \sum\nolimits_{k \in \mathcal{K}}\mathbf{q}_{p,k}m_{p,k},
\end{align}
where  $\mathbf{q}_c = [\mathbf{q}_{c,1}^H \dotsi \mathbf{q}_{c,S}^H]^H \in \mathbb{C}^{MS \times 1}$ and $\mathbf{q}_{p,k} = [\mathbf{q}_{p,k,1}^H \dotsi \mathbf{q}_{p,k,S}^H]^H \in \mathbb{C}^{MS \times 1}$.
Each satellite $s$ should satisfy a per-satellite transmit power constraint stated as
\begin{align}
    p_s(\mathbf{Q}) = \text{tr}\left( 
 \mathbf{E}_s^H \left( \mathbf{q}_c\mathbf{q}_c^H + \sum\nolimits_{k\in\mathcal{K}}\mathbf{q}_{p,k}\mathbf{q}_{p,k}^H \right) \mathbf{E}_s\right)
    \leq P,
\end{align} 
where $P$ is the transmit power budget per satellite, and $\mathbf{E}_s \in \mathbb{C}^{MS \times M}$ is a selection matrix consisting of zeros except for the rows from $M(s-1) + 1$ to $Ms$, which form an identity matrix. 

UT $k$ decodes the common data stream $m_c$ and its own private stream $m_{p,k}$ using successive interference cancellation (SIC) decoding.
Unlike the private streams, the common stream should be decoded by all UTs.
Accordingly, the ergodic rates of the common and private streams are given by $R_c = \min_{k \in \mathcal{K}} f_{c,k}(\mathbf{Q})$ and $R_{p,k} = f_{p,k}(\mathbf{Q})$, respectively, where
\begin{align}
    &f_{c,k}(\mathbf{Q}) = \mathbb{E}_{\{\hat{\mathbf{d}}_{k,s}\}_{s \in \mathcal{S}}} 
 \left[\log_2 \det \left(\mathbf{I} +  \mathbf{\Sigma}_{c,k}^{-1} \mathbf{H}_k \mathbf{q}_c\mathbf{q}_c^H \mathbf{H}_k^H\right)\right], \label{eq:ergodic-rate}\\
    &f_{p,k}(\mathbf{Q}) = \mathbb{E}_{\{\hat{\mathbf{d}}_{k,s}\}_{s \in \mathcal{S}}} 
 \left[\log_2\det \left(\mathbf{I} + \mathbf{\Sigma}_{p,k}^{-1} \mathbf{H}_k \mathbf{q}_{p,k}\mathbf{q}_{p,k}^H \mathbf{H}_{k}^H\right) \right], \nonumber 
\end{align}
with $\mathbf{\Sigma}_{c,k} = \sum_{l\in\mathcal{K}} \mathbf{H}_k \mathbf{q}_{p,l}\mathbf{q}_{p,l}^H \mathbf{H}_k^H + \sigma_z^2\mathbf{I}$ and $\mathbf{\Sigma}_{p,k} = \sum_{l \in \mathcal{K} \setminus \{k\}} \mathbf{H}_k\mathbf{q}_{p,l} \mathbf{q}_{p,l}^H \mathbf{H}_k^H + \sigma_z^2\mathbf{I}$.
The function $f_{c,k}(\mathbf{Q})$ represents the maximum ergodic rate at which the common stream $m_c$ can be decoded by UT $k$.
The overall data rate for each UT $k$ is given by $R_k = R_{c,k} + R_{p,k}$, where $R_{c,k} \geq 0$ is the portion of the common data rate allocated to UT $k$, subject to the constraint $\sum_{k \in \mathcal{K}}R_{c,k}=R_c$.

\section{Proposed sCSI-Aware Optimization} \label{sec:proposed-sCSI-aware-optimization}

To ensure the fairness among UTs, we aim to optimize the precoding matrix $\mathbf{Q}$ to maximize the MMFR metric\footnote{We note that our approach can straightforwardly be applied to maximizing the weighted sum-rate $\sum_{k\in\mathcal{K}}R_k$, which is another widely used fairness-promoting metric.}, defined as $\min_{k \in \mathcal{K}}R_k$ (i.e., the minimum ergodic rate across all UTs), while satisfying the power constraints at all LEO satellites. 
This optimization problem can be formulated as
\begin{subequations} \label{eq:problem-original}
\begin{align}
    \underset{\mathbf{Q}, \mathbf{R}}{\mathrm{max.}}\,\,\, &\min_{k \in \mathcal{K}} \left(R_{c,k} + R_{p,k}\right) \label{eq.problem-objective function} \\
    \mathrm{s.t. }\,\,\,\, & \sum\nolimits_{l \in \mathcal{K}} R_{c,l} \leq f_{c,k}\left(\mathbf{Q}\right), \, k\in\mathcal{K}, \label{eq:problem-original-common-rate}\\
    & R_{p,k} \leq f_{p,k}\left(\mathbf{Q}\right), \, k \in \mathcal{K}, \label{eq:problem-original-private-rate}\\
    & p_s\left(\mathbf{Q}\right) \leq P,\, s\in\mathcal{S}, \label{eq:problem-original-power} \\
    & \mathbf{q}_{p,k,s} = \mathbf{0},\,\, \mathrm{if} \,\,k \notin \mathcal{K}_s, \, s\in \mathcal{S}, \label{eq:problem-original-zero}
\end{align}
\end{subequations}
where $\mathbf{R} = \{R_{j,l} \}_{j\in \{c,p\},l\in\mathcal{K}}$.
The constraint (\ref{eq:problem-original-zero}) ensures that each satellite $s$ has access to the private data streams intended only for its associated UTs $\mathcal{K}_s$.

We note that the rate functions $f_{c,k}(\mathbf{Q})$ and $f_{p,k}(\mathbf{Q})$ do not admit closed-form expressions and typically require computationally intensive sample-mean approximations of ergodic rates.
To overcome this difficulty, we approximate the problem using tractable closed-form upper bounds on the ergodic rates, as presented in Proposition 1.

\begin{prop}[Ergodic Rate Upper Bounds]
The ergodic rate functions of the common and private data streams, $f_{c,k}(\mathbf{Q})$ and $f_{p,k}(\mathbf{Q})$ in (\ref{eq:ergodic-rate}), are upper bounded as $f_{c,k}(\mathbf{Q}) \leq f_{c,k}^{\mathrm{UB}}(\mathbf{Q})$ and $f_{p,k}(\mathbf{Q}) \leq f_{p,k}^{\mathrm{UB}}(\mathbf{Q})$, respectively, where
\begin{subequations} \label{eq:upper-bounds}
\begin{align}
    & f_{c,k}^{\mathrm{UB}}(\mathbf{Q}) = \log_2\det\left(\mathbf{I} + \left(\mathbf{\Sigma}_{c,k}^{\mathrm{UB}}\right)^{-1} \hat{\mathbf{H}}_k \mathbf{q}_c\mathbf{q}_c^H \hat{\mathbf{H}}_k^H\right),\label{eq:upper-bound-common} \\
    & f_{p,k}^{\mathrm{UB}}(\mathbf{Q}) = \log_2\det\left(\mathbf{I} + \left(\mathbf{\Sigma}_{p,k}^{\mathrm{UB}}\right)^{-1} \hat{\mathbf{H}}_k \mathbf{q}_{p,k}\mathbf{q}_{p,k}^H \hat{\mathbf{H}}_k^H\right).\label{eq:upper-bound-private}
\end{align}
\end{subequations}
Here, we have defined $\mathbf{\Sigma}_{c,k}^{\mathrm{UB}} = \sum_{l \in \mathcal{K}}\hat{\mathbf{H}}_k \mathbf{q}_{p,l}\mathbf{q}_{p,l}^H \hat{\mathbf{H}}_k^H + \sigma_z^2\mathbf{I}$, $\mathbf{\Sigma}_{p,k}^{\mathrm{UB}} = \sum_{l \in \mathcal{K} \setminus \{k\}}\hat{\mathbf{H}}_k \mathbf{q}_{p,l}\mathbf{q}_{p,l}^H \hat{\mathbf{H}}_k^H + \sigma_z^2\mathbf{I}$, $\hat{\mathbf{H}}_k = (\mathbf{G}_k \hat{\mathbf{D}}_k\mathbf{G}_k^H)^{\frac{1}{2}}$ with $\hat{\mathbf{D}}_{k}=\mathbb{E}[\mathbf{D}_k^H \mathbf{D}_k]$ given as a function of sCSI as $\hat{\mathbf{D}}_{k}=\mathrm{diag}(\{\beta_{k,s}\}_{s \in \mathcal{S}}) + \mathbf{D}_k^{\mathrm{LoS}}$, and $\mathbf{D}_k^{\mathrm{LoS}}$ is given by $[\mathbf{D}_{k}^{\mathrm{LoS}}]_{s',s}=\sqrt{\frac{\kappa_{k,s'}\beta_{k,s'}\kappa_{k,s}\beta_{k,s}}{(\kappa_{k,s'}+1)(\kappa_{k,s}+1)}}\mathbf{d}_{k,s',0}^H \mathbf{d}_{k,s,0}$ for $s' \neq s$, while its diagonal elements are zero.
\end{prop}

\begin{proof}
To prove $f_{c,k}(\mathbf{Q}) \leq f_{c,k}^{\mathrm{UB}}(\mathbf{Q})$, we start from
\begin{subequations} \label{eq:proof}
\begin{align}
    &\mathbb{E}\left[\log_2 \det \left(\mathbf{I} + \mathbf{\Sigma}_{c,k}^{-1} \mathbf{H}_k \mathbf{q}_c\mathbf{q}_c^H \mathbf{H}_k^H \right)\right] \label{eq:proof-1} \\
    &=\mathbb{E}\left[\log_2 \det \left(\sigma_z^2 \mathbf{I} +  \left(\mathbf{q}_c\mathbf{q}_c^H + \sum_{l \in \mathcal{K}} \mathbf{q}_{p,l}\mathbf{q}_{p,l}^H\right) \mathbf{H}_k^H\mathbf{H}_k\right) \right] \nonumber \\
    &\,\,\,\,\,-\mathbb{E}\left[\log_2 \det \left(\sigma_z^2 \mathbf{I} + \sum\nolimits_{l \in \mathcal{K}} \mathbf{q}_{p,l} \mathbf{q}_{p,l}^H\mathbf{H}_k^H\mathbf{H}_k\right)\right] \label{eq:proof-3} \\
    & \stackrel{\text{(a)}}{\leq} \log_2\det \left( \sigma_z^2\mathbf{I} + \left(\mathbf{q}_c\mathbf{q}_c^H + \sum_{l \in \mathcal{K}} \mathbf{q}_{p,l}\mathbf{q}_{p,l}^H \right) \mathbb{E} \left[\mathbf{H}_k^H\mathbf{H}_k \right]\right) \nonumber \\
    &\,\,\,\,\,-\log_2\det \left(\sigma_z^2\mathbf{I} + \sum\nolimits_{l \in \mathcal{K}} \mathbf{q}_{p,l}\mathbf{q}_{p,l}^H \mathbb{E} \left[\mathbf{H}_k^H\mathbf{H}_k \right] \right), \label{eq:proof-4} 
\end{align}
\end{subequations}
where the expectation is taken over $\{\hat{\mathbf{d}}_{k,s}\}_{s \in \mathcal{S}}$, and (a) follows from the concavity of the function $f(\mathbf{X}) = \log_2\det(\mathbf{I} + \mathbf{A}\mathbf{X}) - \log_2\det(\mathbf{I} + \mathbf{B}\mathbf{X})$ for $\mathbf{A}-\mathbf{B}\succeq \mathbf{0}$ \cite[Lem. 2]{Sun:TC15}, and by applying Jensen's inequality. 
By substituting $\mathbb{E}_{\{\hat{\mathbf{d}}_{k,s}\}_{s \in \mathcal{S}}} [\mathbf{H}_k^H\mathbf{H}_k ] = \hat{\mathbf{H}}_k^H\hat{\mathbf{H}}_k$ into the right-hand side (RHS) of (\ref{eq:proof-4}), we obtain the desired upper bound $f_{c,k}(\mathbf{Q}) \leq f_{c,k}^{\mathrm{UB}}(\mathbf{Q})$.
The other upper bound $f_{p,k}(\mathbf{Q}) \leq f_{p,k}^{\mathrm{UB}}(\mathbf{Q})$ follows similarly by replacing the interference terms accordingly.
\end{proof}

Note that the upper bounds in (\ref{eq:upper-bounds}) are provided in closed-form expressions and do not involve the expectation operator.
Using these upper bounds, the original problem (\ref{eq:problem-original}) is approximated as\footnote{Since the MMFR optimized from (\ref{eq:problem-ub}) is not guaranteed to be achievable, the ergodic rate functions $f_{c,k}(\mathbf{Q})$ and $f_{p,k}(\mathbf{Q})$ should be evaluated after optimization to assess the actual performance.}\begin{subequations} \label{eq:problem-ub}
\begin{align}
    \max_{\mathbf{Q}, \mathbf{R}}\,\,\, &\min_{k \in \mathcal{K}} \left(R_{c,k} + R_{p,k}\right) \label{eq:problem-ub-objective-function} \\
    \mathrm{s.t. }\,\,\,\, & \sum\nolimits_{l \in \mathcal{K}} R_{c,l} \leq f_{c,k}^{\text{UB}}\left(\mathbf{Q}\right), \, k\in\mathcal{K}, \label{eq:problem-ub-common-rate}\\
    & R_{p,k} \leq f_{p,k}^{\text{UB}}\left(\mathbf{Q}\right), \, k \in \mathcal{K}, \label{eq:problem-ub-private-rate}\\
    & \mathrm{(\ref{eq:problem-original-power}), (\ref{eq:problem-original-zero})}. \label{eq:problem-ub-power-zero}
\end{align}
\end{subequations}
Although the reformulated problem (\ref{eq:problem-ub}) is more tractable than (\ref{eq:problem-original}), it remains non-convex due to the constraints (\ref{eq:problem-ub-common-rate}) and (\ref{eq:problem-ub-private-rate}).
In what follows, we propose a WMMSE-based algorithm to address this challenge.

The upper bound rate function $f_{c,k}^{\text{UB}}(\mathbf{Q})$ in (\ref{eq:upper-bound-common}) measures the achievable rate of the common data stream $m_c$ when decoded from the virtual received signal $\hat{\mathbf{y}}_{c,k} = \hat{\mathbf{H}}_k \mathbf{q}_c m_c + \sum_{l\in\mathcal{K}} \hat{\mathbf{H}}_k\mathbf{q}_{p,l} m_{p,l} + \mathbf{z}_k$, where $\hat{\mathbf{H}}_k$ is a known sCSI-dependent matrix, as defined in Proposition 1. Similarly, the private upper bound rate $f_{p,k}^{\text{UB}}(\mathbf{Q})$ in (\ref{eq:upper-bound-private}) corresponds to decoding the private stream $m_{p,k}$ from $\hat{\mathbf{y}}_{p,k} = \hat{\mathbf{H}}_k \mathbf{q}_{p,k} m_{p,k} + \sum_{l\in\mathcal{K}\setminus\{k\}} \hat{\mathbf{H}}_k \mathbf{q}_{p,l} m_{p,l} + \mathbf{z}_k$.

With the observation above, we address problem (\ref{eq:problem-ub}) using the WMMSE approach. Specifically, we define the mean squared error (MSE) functions for the common and private data streams as
\begin{subequations}   
\begin{align}
    e_{c,k}\left(\mathbf{u}_{c,k}, \mathbf{Q}\right) &= \mathbb{E}_{\mathbf{m}, \mathbf{z}_k}\left[|m_c - \mathbf{u}_{c,k}^H\hat{\mathbf{y}}_{c,k}|^2\right] = \big|1 - \mathbf{u}_{c,k}^H \hat{\mathbf{H}}_k \mathbf{q}_c \big|^2 + \mathbf{u}_{c,k}^H \mathbf{\Sigma}_{c,k}^{\text{UB}} \mathbf{u}_{c,k}, \\
    e_{p,k}\left(\mathbf{u}_{p,k}, \mathbf{Q}\right) &= \mathbb{E}_{\mathbf{m}, \mathbf{z}_k}\left[|m_{p,k} - \mathbf{u}_{p,k}^H\hat{\mathbf{y}}_{p,k}|^2\right] = \big|1 - \mathbf{u}_{p,k}^H \hat{\mathbf{H}}_k \mathbf{q}_{p,k}\big|^2 + \mathbf{u}_{p,k}^H \mathbf{\Sigma}_{p,k}^{\text{UB}} \mathbf{u}_{p,k},
\end{align}
\end{subequations}
where $\mathbf{u}_{c,k} \in \mathbb{C}^{MS \times 1}$ and $\mathbf{u}_{p,k} \in \mathbb{C}^{MS \times 1}$ represent combining vectors.

Using the MSE expressions, the rate constraints in (\ref{eq:problem-ub-common-rate}) and (\ref{eq:problem-ub-private-rate}) can be equivalently rewritten as
\begin{subequations} \label{eq:constraint-wmmse}
\begin{align}
    &\sum\nolimits_{l\in\mathcal{K}}R_{c,l}\ln2 \leq \max_{v_{c,k}>0, \mathbf{u}_{c,k}} \bigg\{\ln v_{c,k} -  v_{c,k}e_{c,k} \left(\mathbf{u}_{c,k},\mathbf{Q}\right)  + 1\bigg\}, \,\, k\in\mathcal{K}, \label{eq:constraint-wmmse-common} \\
    &R_{p,k} \ln2 \leq \max_{v_{p,k}>0, \mathbf{u}_{p,k}} \bigg\{\ln v_{p,k} - v_{p,k}e_{p,k} \left(\mathbf{u}_{p,k}, \mathbf{Q}\right) + 1 \bigg\}, \,\, k\in\mathcal{K},\label{eq:constraint-wmmse-private}
\end{align}
\end{subequations}
where $v_{c,k}$ and $v_{p,k}$ are weight variables associated with the respective MSE functions.
The optimal combining vectors and weight values, that maximize the RHSs of (\ref{eq:constraint-wmmse}), for given $\mathbf{Q}$ are given by
\begin{subequations} \label{eq:optimal-combiner-weight}
\begin{align}
    &\mathbf{u}_{c,k} =  \left(\hat{\mathbf{H}}_k\mathbf{q}_c\mathbf{q}_c^H\hat{\mathbf{H}}_k^H + \mathbf{\Sigma}_{c,k}^{\text{UB}}\right)^{-1} \hat{\mathbf{H}}_k\mathbf{q}_c, \label{eq:optimal-combiner-common} \\
    &\mathbf{u}_{p,k} = \left(\hat{\mathbf{H}}_k\mathbf{q}_{p,k}\mathbf{q}_{p,k}^H\hat{\mathbf{H}}_k^H + \mathbf{\Sigma}_{p,k}^{\text{UB}}\right)^{-1} \hat{\mathbf{H}}_k\mathbf{q}_{p,k}, \label{eq:optimal-combiner-private} \\
    &v_{c,k} = \frac{1}{e_{c,k} \left(\mathbf{u}_{c,k}, \mathbf{Q}\right)},\,\,
    v_{p,k} = \frac{1}{e_{p,k}\left(\mathbf{u}_{p,k}, \mathbf{Q}\right)}. \label{eq:optimal-weight}
\end{align}
\end{subequations}
Accordingly, problem (\ref{eq:problem-ub}) can be equivalently restated as 
\begin{subequations} \label{eq:problem-ub-convexified}
\begin{align} 
    \max_{\mathbf{Q}, \mathbf{R}, \mathbf{u}, \mathbf{v}}\, &\min_{k \in \mathcal{K}} \left(R_{c,k} + R_{p,k}\right) \\
    \mathrm{s.t. }\,\,\, & \mathrm{(\ref{eq:constraint-wmmse-common}), (\ref{eq:constraint-wmmse-private})}, \mathrm{(\ref{eq:problem-original-power}), (\ref{eq:problem-original-zero})}, \\
    & v_{c,k}>0, v_{p,k}>0, \,\, k\in\mathcal{K},
\end{align}
\end{subequations}
where $\mathbf{u}=\{\mathbf{u}_{j,l}\}_{j\in\{c,p\},l\in\mathcal{K}}$ and $\mathbf{v} = \{v_{j,l}\}_{j\in\{c,p\},l\in\mathcal{K}}$.
Note that for fixed auxiliary variables $\{\mathbf{u},\,\mathbf{v}\}$, the optimization over the beamforming matrix $\mathbf{Q}$ in (\ref{eq:problem-ub-convexified}) becomes convex and can be efficiently solved using standard convex solvers such as CVX \cite{Grant:CVX20}. 
Conversely, for a fixed $\mathbf{Q}$, the optimal $\{\mathbf{u}, \mathbf{v}\}$ can be obtained in closed form as in (\ref{eq:optimal-combiner-weight}).
Therefore, the objective function can be monotonically improved by alternately updating the block variables $\mathbf{Q}$ and $\{\mathbf{u}, \mathbf{v}\}$.
The detailed procedure of the WMMSE algorithm is summarized in Algorithm \ref{algorithm-1}.

\begin{algorithm} 
\caption{\label{algorithm-1}Proposed sCSI-aware WMMSE algorithm}
\textbf{\footnotesize{}1}\textbf{ Initialize:}  Initialize $\mathbf{Q}$ such that (\ref{eq:problem-original-power}) and (\ref{eq:problem-original-zero}) are satisfied; \\
\textbf{\footnotesize{}2}\textbf{ Repeat} \\
\textbf{\footnotesize{}3}~~~update $\{\mathbf{u}, \mathbf{v}\}$ as per (\ref{eq:optimal-combiner-weight}) with $\mathbf{Q}$ fixed; \\
\textbf{\footnotesize{}4}~~~update $\{\mathbf{Q},\mathbf{R}\}$ as a solution of (\ref{eq:problem-ub-convexified}) with $\{\mathbf{u}, \mathbf{v}\}$ fixed; \\
\textbf{\footnotesize{}5} \textbf{Until convergence}
\end{algorithm}

The per-iteration complexity of Algorithm 1 is primarily dominated by that of solving the convex problem in Step 4 of each iteration, which is upper bounded by $\mathcal{O}( n_V (n_V^3 + n_C) \log \epsilon )$ \cite[p. 4]{BTal:LN19}, where $n_V$ denotes the number of optimization variables, $n_C$ represents the number of arithmetic operations required to compute the objective and constraint function, and $\epsilon$ is the desired error tolerance level.
For the convex problem in Step 4, we have $n_V = \mathcal{O}(KMS)$ and $n_C = \mathcal{O}(K^2M^3S^3)$. Additionally, we numerically confirmed that Algorithm 1 consistently converges within a few iterations across all simulated cases.

\section{Numerical Results} \label{sec:numerical}

This section presents numerical results to validate the effectiveness of the proposed algorithm. 
Each satellite operates at an altitude of $H = 600$ km and has a maximum nadir angle $\vartheta_{max}=30^\circ$, resulting in overlapping service areas.
The system operates at a carrier frequency of 2 GHz with a bandwidth of $50$ MHz, and the noise temperature is set to $T_n = 290$ K.
We consider antenna configurations of $[M_{\text{x}}, M_{\text{y}}, N_{\text{x}'},N_{\text{y}'}] = [5, 5, 4, 4]$ with antenna spacings being $[d_{\text{x}}, d_{\text{y}}, d_{\text{x}'}, d_{\text{y}'}] = [\lambda, \lambda, \frac{\lambda}{2}, \frac{\lambda}{2}]$.
The antenna gains are set to $[G_{\text{sat}}, G_{\text{ut}}] = [6, 0]$ dBi. 
The Rician factor follows the S-band suburban scenario specified in \cite{3GPP:Release15}.
UTs are randomly distributed within the combined service areas of the satellites. 
The maximum elevation angle $\alpha_{\text{max}}$ is determined by the maximum nadir angle $\vartheta_{\text{max}}$, and the corresponding maximum distance $D_{\text{max}}$ defines the coverage boundary of each satellite.
The association between satellites and UTs is established based on the distance $D_{k,s}$: if UT $k$ falls within the service area of satellite $s$ (i.e., $D_{k,s} \leq D_{\text{max}}$), then UT $k$ and satellite $s$ belong to $\mathcal{K}_s$ and $\mathcal{S}_k$, respectively.
The covariance matrix of the NLoS component $\mathbf{\Sigma}_{k,s}$ is set as $\mathbf{\Sigma}_{k,s} = \mathrm{diag}(\{\mu_n\}_{n=1}^N)$ with $\mu_n = \tilde{\mu}_n/\sum_{l=1}^N\tilde{\mu}_l$, where $\tilde{\mu}_1,\, \dots,\,\tilde{\mu}_N$ are independently sampled from $\mathcal{U}(0,1)$.
The noise variance $\sigma_z^2$ is given by $k_B T_n B$ with the Boltzmann constant $k_B$.

To validate the effectiveness of the proposed sCSI-aware design for the multi-satellite cooperative RSMA transmission scheme, we compare the average MMFR achieved by RSMA and SDMA under both iCSI and sCSI conditions. 
For the sCSI case, we utilize the upper bounds on the ergodic rates for precoding optimization, as detailed in Sec. \ref{sec:proposed-sCSI-aware-optimization}, while the MMFR performance of the resulting precoding matrix is evaluated using the actual ergodic rate functions, computed via sample-mean approximation with 20,000 channel realizations.
As a baseline, we also show the performance of precoders designed solely based on the LoS components, i.e., directional CSI (dCSI).
Additionally, we include the performance of non-cooperative transmission schemes, where each UT is served solely by its nearest satellite.
The performance of SDMA is evaluated by setting the beamforming vector for the common data stream to a zero vector, i.e., $\mathbf{q}_c = \mathbf{0}$.

\begin{figure}
\centering\includegraphics[width=0.9\linewidth]{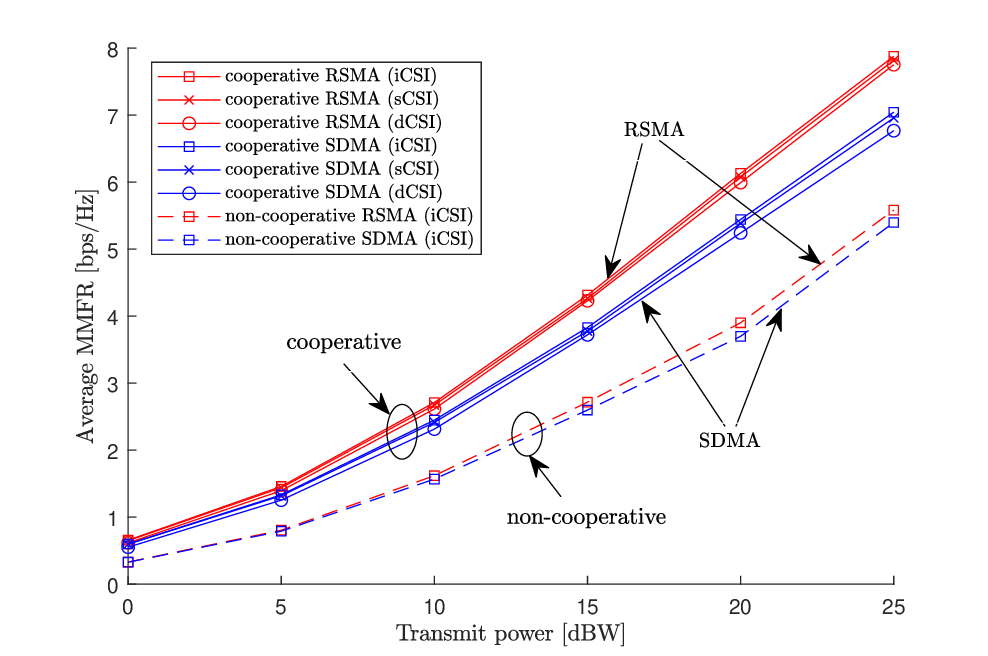}
\caption{\label{fig:Power-MMFR} Average MMFR versus transmit power $P$.}
\end{figure}

Fig. \ref{fig:Power-MMFR} illustrates the average MMFR as a function of satellite transmit power, for a SATCOM system with $S=4$ and $K=6$.
For both cooperative RSMA and SDMA schemes, the sCSI-aware design achieves performance that closely approaches that of the iCSI-based design. This is primarily due to the rank-1 nature of the channel matrices, as described in (\ref{eq:channel-matrix}), where the basis vector $\mathbf{g}_{k,s}$ of the one-dimensional row space of $\mathbf{H}_{k,s}$ is included in the sCSI set.
The performance gain of RSMA over SDMA becomes more pronounced as the transmit power $P$ increases. This is because the signal-to-noise ratio (SNR) at the UTs increases with $P$, and at higher SNRs, interference, which is more effectively handled by RSMA than SDMA, becomes the main performance bottleneck rather than noise.
Moreover, the performance gap between RSMA and SDMA is more significant in the multi-satellite cooperative scenario enabled by ISLs, compared to the non-cooperative transmission case.
This highlights the synergistic benefits of combining RSMA with cooperative satellite transmission strategies.
The performance loss of the dCSI case compared to sCSI is more pronounced at SDMA than RSMA, indicating that the RSMA scheme is more robust to CSI imperfections.

\begin{figure}
\centering\includegraphics[width=0.9\linewidth]{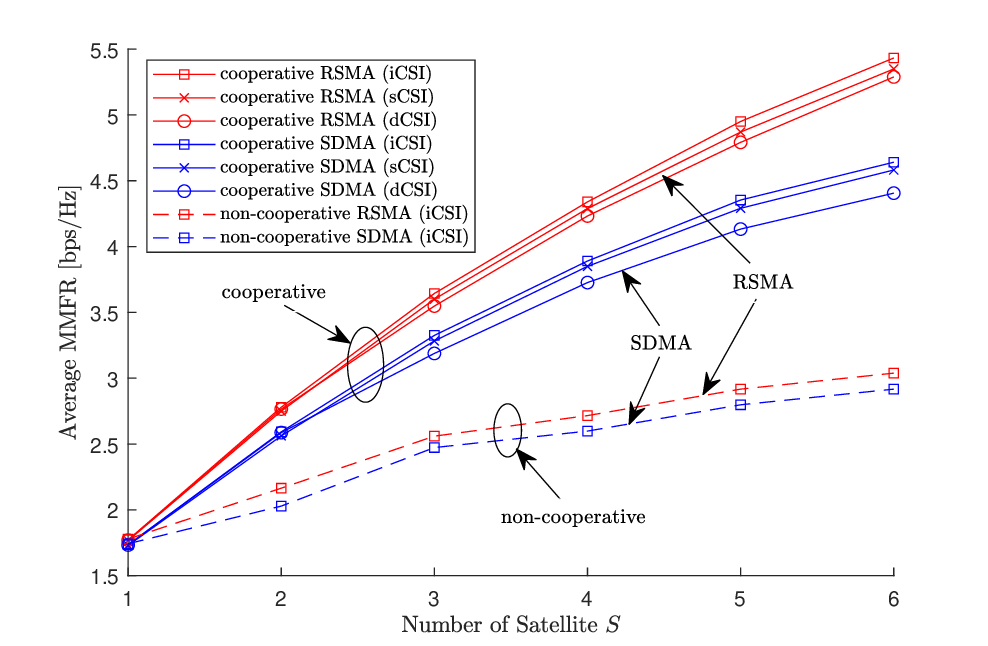}
\caption{\label{fig:Numsat-MMFR} Average MMFR versus the number of satellites $S$.}
\end{figure}

Fig. \ref{fig:Numsat-MMFR} shows the average MMFR while increasing the number of LEO satellites $S$, for a SATCOM system with $K=6$ and $P=15$ dBW.
As the number of satellites increases, each UT is more likely to fall within overlapping coverage areas of multiple satellites. This leads to improved SNR at the UTs, resulting in more pronounced performance gains of RSMA over SDMA.
Furthermore, as satellite density increases with larger $S$, the benefits of inter-satellite cooperation enabled by ISLs become more evident.
These observations support the effectiveness of the proposed approach that integrates RSMA with multi-satellite cooperative transmission, particularly in large-scale satellite constellations such as Starlink and OneWeb.

\section{Conclusion} \label{sec:conclusion}

We have studied the design of mMIMO precoding at LEO satellites for inter-satellite cooperative RSMA transmission.
Considering the challenges in acquiring iCSI in SATCOM, we based the precoding design solely on sCSI.
To handle the difficulty of computing ergodic rates, we approximated the ergodic MMFR maximization problem using closed-form upper bounds and developed a WMMSE-based algorithm.
Numerical results demonstrate that the proposed sCSI-based scheme closely approaches the performance of the iCSI-based counterpart, except in scenarios dominated by NLoS components. Furthermore, the performance gain of RSMA over SDMA becomes more significant when combined with inter-satellite cooperative transmission.


\begin{thebibliography}{1}
\bibliographystyle{IEEEtran}

\bibitem{Bjornson-et-al:TWC20}
E. Bj{\"o}rnson and L. Sanguinetti, ``Making cell-free massive MIMO competitive with MMSE processing and centralized implementation,'' \textit{IEEE Trans. Wireless Commun.}, vol. 19. no. 1, pp. 77--90, Jan. 2020.


\bibitem{You:JSAC20}
L, You, K.-X. Li, W. Gao, X.-G. Xia and B. Ottersten, ``Massive MIMO transmission for LEO satellite communications,'' \textit{IEEE. J. Sel. Areas Commun.}, vol. 38, no. 8, pp. 1851--1865, Aug. 2020.

\bibitem{Li:TC22}
K.-X. Li and \textit{et al.}, ``Downlink transmit design for massive MIMO LEO satellite communications,'' \textit{IEEE Trans. Commun.}, vol. 70, no. 2, pp. 1014--1028, Feb. 2022.

\bibitem{You:JSAC22}
L. You and \textit{et al.}, ``Beam squint-aware integrated sensing and communications for hybrid massive MIMO LEO satellite systems,'' \textit{IEEE. J. Sel. Areas Commun.}, vol. 40, no. 10, pp. 2944--3009, Oct. 2022.

\bibitem{Xiang:TC24}
Z. Xiang, X. Gao, K.-X. Li and X.-G. Xia, ``Massive MIMO downlink transmission for multiple LEO satellite communication,'' \textit{IEEE Trans. Commun.}, vol. 72, no. 6, pp. 3352--3364, Jun. 2024.

\bibitem{Xiang:TAES24}
Z. Xiang and \textit{et al.}, ``Massive MIMO uplink transmission for multiple LEO satellite communication,'' \textit{IEEE Trans. Aerosp. Electron. Syst.}, vol. 61, no. 2, pp. 4852--4865, Apr. 2025.

\bibitem{Zhang:TWC24}
X. Zhang, S. Sun, M. Tao, Q. Huang and X. Tang, ``Multi-satellite cooperative networks: Joint hybrid beamforming and user scheduling design,'' \textit{IEEE Trans. Wireless Commun.}, vol. 23, no. 7, pp. 7938--7952, Jul. 2024.

\bibitem{Abdelsadek:TC23}
M. Y. Abdelsadek and \textit{et al}., ``Future space networks: Toward the next giant leap for humankind,'' \textit{IEEE Trans. Commun.}, vol. 71, no. 2, pp. 949--1007, Feb. 2023.


\bibitem{Park:Net23}
J. Park and \textit{et al}., ``Rate-splitting multiple access for 6G networks: Ten promising scenarios and applications,'' \textit{IEEE Netw.}, vol. 38, no. 3, pp. 128--136, May 2024.

\bibitem{Yin:TC21}
L. Yin and B. Clerckx, ``Rate-splitting multiple access for multigroup multicast and multibeam satellite systems,'' \textit{IEEE Trans. Commun.}, vol. 69, no. 2, pp. 976--990, Feb. 2021.

\bibitem{Xu:TC24}
Y. Xu, L. Yin, Y. Mao, W. Shin and B. Clerckx, ``Distributed rate-splitting multiple access for multilayer satellite communications,'' \textit{IEEE Trans. Commun.}, vol. 72, no. 10, pp. 6131--6144, Oct. 2024.

\bibitem{Liu:TWC24}
Z. Liu, L. Yin, W. Shin and B. Clerckx, ``Rate-splitting multiple access for quantized ISAC LEO satellite systems: A max-min fair energy-efficient beam design,'' \textit{IEEE Trans. Wireless Commun.}, vol. 23, no. 10, pp. 15394--15408, Oct. 2024.

\bibitem{Christensen:TWC08}
S. S. Christensen, R. Agarwal, E. De Carvalho and J. M. Cioffi, ``Weighted sum-rate maximization using weighted MMSE for MIMO-BC beamforming design,'' \textit{IEEE Trans. Wireless Commun.}, vol. 7, no. 12, pp. 4792--4799, Dec. 2008.

\bibitem{Sun:TC15}
C. Sun, X. Gao, M. Matthaiou, Z. Ding and C. Xiao, ``Beam Division Multiple Access Transmission for Massive MIMO Communications,'' \textit{IEEE Trans. Commun.}, vol. 63, no. 6, pp. 2170--2184, Jun. 2015

\bibitem{Grant:CVX20}
M. Grant and S.Boyd, ``CVX: MATLAB software for disciplined convex programming,'' Second Edition, Third Printing, pp. 1--786, Ver. 2.2, Jan. 2020.

\bibitem{BTal:LN19}
A. Ben-Tal and A. Nemirovski, ``Lectures on modern convex optimization,'' Lecture Notes, Georgia Inst. Technol., Atlanta, GA, USA, 2019. [Online]. Available: https://www2.isye.gatech.edu/nemirovs/LMCO$\_$LN.pdf

\bibitem{3GPP:Release15}
3GPP, ``Study on new radio (NR) to support non-terrestrial networks (Release 15),'' TR 38.811, V15.3.0, Jul. 2020.

\end{thebibliography}
\end{document}